\documentclass[journal]{IEEEtran}
\def\BibTeX{{\rm B\kern-.05em{\sc i\kern-.025em b}\kern-.08em
    T\kern-.1667em\lower.7ex\hbox{E}\kern-.125emX}}
\usepackage{amsmath,mathtools,amsfonts,amssymb,amsthm}

\usepackage{array}
\usepackage{cite}
\usepackage{xcolor}
\usepackage{threeparttable}
\usepackage[caption=false,font=normalsize,labelfont=sf,textfont=sf]{subfig}
\usepackage{textcomp}
\usepackage{stfloats}
\usepackage{url}
\usepackage{verbatim}
\usepackage{graphicx}
\usepackage{cuted}
\usepackage{float} 
\usepackage{balance}
\usepackage{booktabs} 
\usepackage{caption}                         
\usepackage{algpseudocode}
\newtheoremstyle{sltheorem}
{}                
{}                
{}        
{10pt}                
{\bfseries}       
{:}               
{ }               
{}                
\theoremstyle{sltheorem}
\newtheorem{theorem}{Theorem}
\newtheorem{definition}{Definition}
\newtheorem{corollary}{Corollary}

\newtheorem{property}{Property}
\newtheorem{assumption}{Assumption}
\begin{document}
\title{Accelerating Chance-constrained SCED via Scenario Compression}
\author{Qian Zhang,~\IEEEmembership{Student Member, IEEE,} Le Xie,~\IEEEmembership{Fellow, IEEE}
\thanks{Qian Zhang and Le Xie are with the School of Engineering and Applied Sciences, Harvard University, Cambridge, MA 02138 USA (e-mail: qianzhang@g.harvard.edu; xie@seas.harvard.edu).}}

\maketitle

\begin{abstract}
This paper studies some compression methods to accelerate the scenario-based chance-constrained security-constrained economic dispatch (SCED) problem. In particular, we show that by exclusively employing the vertices after convex hull compression, an equivalent solution can be obtained compared to utilizing the entire scenario set. For other compression methods that might relax the original solution, such as box compression, this paper presents the compression risk validation scheme to assess the risk arising from the sample space. By quantifying the risk associated with compression, decision-makers are empowered to select either solution risk or compression risk as the risk metric, depending on the complexity of specific problems. Numerical examples based on the 118-bus system and synthetic Texas grids compare these two risk metrics. The results also demonstrate the efficiency of compression methods in both problem formulation and solving processes.
\end{abstract}


\section{Introduction}
\IEEEPARstart{I}{n}creasing penetration of variable renewable energy has posed unprecedented challenges for secure and economic dispatch in large power systems.  Many uncertainty optimization methods are proposed to replace the conventional deterministic SCED, especially \textit{stochastic optimization} (SO) \cite{prekopa2013stochastic} and \textit{robust optimization} (RO) \cite{ben2009robust}. Chance-constrained optimization (CCO), which combines the advantages of both SO and RO, has found a high potential application in SCED and many other areas in recent years \cite{geng2019data}. TABLE \ref{tab:taxonomy} compares the strengths and weaknesses of SO, RO, and CCO with select examples in the power system SCED. Unlike optimizing under expected value (SO) or making solutions robust in the whole uncertainty set (RO), CCO gives the feasibility guarantee of solutions under \emph{user-defined} probability. This feature allows decision-makers to balance risk and cost according to their preferences. 
\begin{table*}[h]
\centering
\begin{tabular}{>{\centering\arraybackslash\color{black}}p{2cm}>{\centering\arraybackslash\color{black}}p{4cm}>{\centering\arraybackslash\color{black}}p{4cm}>{\centering\arraybackslash\color{black}}p{4cm}}
\toprule
\textbf{Method} & \textbf{Advantages} & \textbf{Disadvantages} & \textbf{Power System Applications} \\ \midrule
\textbf{SO} & Optimizes expected value; flexible with various uncertainties & Computationally expensive; poor performance in worst scenarios & \cite{phan2014two,gu2016stochastic,kannan2020stochastic} \\ \midrule
\textbf{RO} & Less sensitive to uncertainty; computationally efficient & Overly conservative; simplifies uncertainty to bounded ranges & \cite{lorca2014adaptive,lubin2015robust,xie2013short} \\ \midrule
\textbf{CCO} & Realistic modeling of uncertainty; balances risk and cost & Computationally expensive; sensitive to distribution accuracy & \cite{bienstock2014chance,modarresi2018scenario,yang2021chance} \\ \bottomrule
\end{tabular}
\caption{{\color{black}Comparison of SO, RO, and CCO and some selected examples in power system SCED}}
\label{tab:taxonomy}
\end{table*}

\indent To make the CCO problem computationally tractable, there are two main approaches to modeling the uncertainty. A brute-force way is to truncate the uncertainty space based on the assumed probability distribution function and a tolerable risk level \cite{bertsimas2006tractable}, which has limitations when the probability distribution is complicated or hard to validate, such as the wind power forecasting error \cite{bludszuweit2008statistical,miettinen2020simulating}. In contrast,  data-driven methods describing the uncertainty directly from sampling are more practical without any assumptions on the probability distribution \cite{geng2019data}. This paper  focuses on adopting a method that is under active development  called the \textit{scenario approach}, which provides probabilistic guarantees based on a finite number of samples \cite{calafiore2006scenario,campi2021scenario}.\\
\indent In the development of the scenario approach, most research aims to find the precise relationship between the sample size and the risk level with a rigorous mathematical guarantee. The first exact expression of this relationship is presented in \cite{campi2008exact}, where the problem is regarded as the so-called \textit{fully supported} to give the risk bound \textit{before} solving the problem. After observing the fact that many reality optimization problems are not fully supported, the following works prefer to give a more tight risk bound \textit{after} solving the problem based on posterior evaluation of the support constraints (or complexity) \cite{campi2018wait}. Some recent works extend the above results to non-convex or nonstationary models \cite{campi2018general,shukla2024sample} and some machine learning applications \cite{campi2021theory,campi2023compression}.\\
\indent Along with the continuous breakthrough on the theory side, many papers spontaneously try to make the scenario approach computationally efficient mainly in two directions: saving the problem-solving time or simplifying the problem formulation. Except for advanced optimization algorithms which are updated in commercial solvers, the solving time in scenario approach relies on finding the problem's complexity after gaining the solution. In \cite{geng2021computing}, some algorithms are proposed to compute complexity in both convex and non-convex models under different problem structure assumptions. To find the complexity with minimum data resources, the authors introduce a new scheme where the size of input scenarios is incrementally tuned to hit a desired level of risk \cite{garatti2022complexity}. On the other hand, simplifying the scenario approach formulation is also a good way to improve both computational efficiency and memory. Margellos \textit{et al} \cite{margellos2014road} convert the scenario approach problem to a two-part robust version, where the scenario-based uncertainty is relaxed to set-based uncertainty first. This new formulation has the same risk guarantee and less computation time but performs worse than the original solution. {\color{black}Because the mathematical formulation of scenario approach is equal to applying robust philosophy to each sample of scenarios, some scenario reduction methods in chance-constrained or robust optimization areas \cite{shang2019data,wang2020data,goerigk2023optimal,jiang2024mathematical} inspire the development of a more general scenario compression scheme based on \emph{scenario approach} \cite{margellos2015connection,campi2023compression}.}\\
\indent These above works have found many successful applications in power systems with different sources of uncertainty: security assessment \cite{vrakopoulou2013probabilistic}, demand response \cite{ming2017scenario}, economic dispatch \cite{modarresi2018scenario}, energy storage planning \cite{yan2022two} and unit commitment \cite{geng2021computing}, etc. Among all these applications, the {\color{black}SCED} process is the area that has both strict security requirements and solving time constraints, typically within 5 minutes. To achieve the optimal solution of the original scenario approach formulation, the previous scenario-based economic dispatch papers \cite{geng2019data,ming2017scenario,modarresi2018scenario} refuse to simplify the problem formulation. Because of the convex property, the problem-solving algorithms in commercial solvers have a polynomial time, but the problem formulation process consumes a long time and huge computer memory in large-scale power grids which typically requires high-performance computing \cite{modarresi2018scenario}. Meanwhile, a simplified problem formulation can also save problem-solving time. In this paper, we try to make this problem formulation process more efficient for {\color{black}SCED} and our main contributions are twofold.\\
\indent First, the scenarios compression method is proposed to simplify the scenario-based {\color{black}SCED} formulation without any uncertainty space relaxation. After finding the convex hull of the scenarios-based uncertainty space, we prove the equivalent of inputting the vertexes and the whole scenarios. Unlike too general problem setting in \cite{margellos2014road}, our result mainly relies on the special structure of the {\color{black}SCED} problem, where the decision variables and the uncertainty variables are just linearly combined with each other.\\
\indent Second, leveraging on the recent theory findings given in \cite{campi2023compression,garatti2022risk}, we can also estimate both the upper and lower compression risk bound based on the number of vertexes. Different from evaluating the risk based on the complexity of the solution, this new validation scheme directly generates the compression risk based on the complexity of the sample space, which quantifies the probability of the new scenario violating the compressed space. {\color{black}While the solution risk bounds are provided after solving the optimization problem, the compression risk bounds can be directly established in a \emph{prior} manner. In practice, engineers can choose from these risk-bound options based on their preference for either accuracy or computational efficiency.}\\
\indent The remainder of this article is organized as follows. Section \ref{sec:Problem} formulates the chance-constrained {\color{black}SCED} problem and introduces the \textit{updated} theory results of scenario approach in solving chance-constrained problems. After clarifying the main computational obstacles of applying scenario approach to a large-scale system at the end of Section \ref{sec:Problem}, Section \ref{sec:compre} proposes two useful scenario compression methods to save the problem formulation and solving time. The new compression risk validation scheme is also presented to quantify the risk from sample space rather than solution space. The efficacy of the proposed approach is demonstrated on 118-bus and synthetic Texas grids in Section~\ref{sec:case}. Section \ref{sec:conclusion} gives the concluding remarks.

\section{Problem Formulation} \label{sec:Problem}
\subsection{Chance-Constrained {\color{black}SCED}}
\indent {\color{black} The conventional {\color{black}SCED} problem is reformulated in a chance-constrained form, which is also similar with the chance-constrained DC optimal power flow (DC-OPF) in many previous papers \cite{vrakopoulou2013probabilistic,bienstock2014chance}.} In our setting, system uncertainty refers to the wind power forecasting error, which is the main uncertain factor in the economic dispatch process. This model can also be extended to other uncertain sources, such as solar and demand, {\color{black}where the uncertainty can be assumed either following a distribution function \cite{ming2017scenario} or directly generated from historical data \cite{zhang}. }\\
\indent The wind generation $w=\hat{w}+\tilde{w}$ consists of \textit{deterministic} wind forecast value $\hat{w} \in \mathbf{R}^{n_w}$ and the \textit{uncertain} forecast error $\tilde{w} \in \Delta$, where $\Delta\subseteq\mathbf{R}^{n_w}$ is the uncertainty set. After introducing wind power into the conventional chance-constrained DC-OPF problem, we have:
\begin{subequations} 
\label{cco}
\begin{align}
\min _{g, \eta}\; &c(g)   \label{subeqn:obj}\\
\text { s.t. }&  \mathbf{1}^{\top} g=\mathbf{1}^{\top} d-\mathbf{1}^{\top} \hat{w} \label{1b} \\
& \underline{g} \preceq g \preceq \bar{g} \label{1c}\\
& f(\hat{w}, \tilde{w})=H_g\left(g-\mathbf{1}^{\top} \tilde{w} \eta\right)-H_d d+H_w(\hat{w}+\tilde{w}) \label{subeqn:f}\\
& \mathbb{P}_{\tilde{w}}\left(\begin{array}{l}\underline{f} \preceq f(\hat{w}, \tilde{w}) \preceq \bar{f}\\
\underline{g} \preceq g-\mathbf{1}^{\top} \tilde{w} \eta \preceq \bar{g}\\
R_d\preceq-\mathbf{1}^{\top} \tilde{w} \eta\preceq R_u \end{array}\right) \geq 1-\epsilon \label{subeqn:risk}\\
& \mathbf{1}^{\top} \eta=1 \label{1f}
\end{align}
\end{subequations}
where the decision variables are generation output levels $g \in \mathbf{R}^{n_g}$, and an affine control policy $\eta \in \mathbf{R}^{n_g}$ of automatic generation control to allocate total wind fluctuation in real-time. The objective function is the total generations cost $c(g)$. The load level is $d \in \mathbf{R}^{n_d}$, and transmission line flows $f \in \mathbf{R}^{n_f}$ are calculated using (\ref{subeqn:f}), where $H_g$, $H_d$, and $H_w $ are the corresponding sub-matrix of the power transfer distribution factor (PTDF) matrix $H$. {\color{black}The constraints related to the forecasting error $\tilde{w}$ are modeled as a chance-constraint form under risk $\epsilon$ in (\ref{subeqn:risk}), including transmission line flow limits $[\underline{f}, \bar{f}]\in \mathbf{R}^{n_f} \times \mathbf{R}^{n_f}$, generation capacity limits $[\underline{g}, \bar{g}]\in \mathbf{R}^{n_g} \times \mathbf{R}^{n_g}$ and the ramp up(down) rate limits $[R_d, R_u]\in \mathbf{R}^{n_g} \times \mathbf{R}^{n_g}$.}\\

\indent Because of the probability measure term $\mathbb{P}_{\tilde{w}}$, it is impossible to solve the above problem directly. Over the past decade, many different ways have been proposed to reformulate chance-constrained problems into tractable deterministic optimization problems, such as sample average approximation (SAA) \cite{luedtke2008sample} and moment-based reformulations \cite{roald2015security}. Interested readers may refer to the latest review \cite{roald2023power} for a more comprehensive summary.
\subsection{The Scenario Approach}
\indent In this paper, we focus on a novel data-driven method called \textit{scenario approach} \cite{campi2008exact} to make the above chance-constrained problem tractable. Leveraging on the randomization of the sampled scenarios, the scenario approach provides probabilistic guarantees based on the finite identical independent distributed (\textit{i.i.d.}) scenarios. For the chance-constrained {\color{black}SCED} problem (\ref{cco}), supposing we have the \textit{i.i.d.} wind forecasting error scenarios set $\mathcal{N}:=\left\{\tilde{w}_1, \tilde{w}_2, \cdots, \tilde{w}_N\right\}$, the chance-constrained inequalities (\ref{subeqn:risk}) can be replaced by scenario-based inequalities (\ref{saa})-(\ref{sac}):
\begin{subequations} \label{saaa}
\begin{align}
\min _{g, \eta}\; &c(g)   \nonumber\\
\text { s.t }&  \underline{f} \preceq f(\hat{w}, \tilde{w}_i) \preceq \bar{f}\label{saa}\\
&\underline{g} \preceq g-\mathbf{1}^{\top} \tilde{w}_i \eta \preceq \bar{g}\label{sab}\\
&R_d\preceq-\mathbf{1}^{\top} \tilde{w}_i \eta\preceq R_u \label{sac}\\ 
& i=1,2,3,...,N \nonumber \\
& (\ref{1b}), (\ref{1c}), (\ref{subeqn:f}), (\ref{1f}) \nonumber
\end{align}
\end{subequations}
\indent Before giving the quantified risk level for the above scenario problem $\mathrm{SP}(\mathcal{N})$, we introduce some basic definitions under the scenario approach scheme.
\begin{definition}[Support Constraint] \label{supportconst}
The scenario-dependent constraint corresponding to sample $\tilde{w}_s, s \in \{1,2,...,\mathcal{S}\}$, is a \textit{support constraint} or \textit{support scenario} if its removal improves the solution of $\mathrm{SP}(\mathcal{N})$, i.e., if it decreases the optimal cost $c(g)$.    
\end{definition}
\begin{assumption}[Non-degeneracy] 
For every $N$, the solution $(g^*,\eta^*)$ to the optimization problem (\ref{saaa}) coincides with probability 1 with the solution that is obtained after eliminating all the constraints that are not of support.    
\end{assumption}
\begin{definition}[Solution Complexity]
The number of support scenarios in $\mathrm{SP}(\mathcal{N})$ is defined as the solution complexity $s_N^*$.    
\end{definition}
\textit{Remark}: Rigorous speaking, the number of support scenarios is related to both solution and sample distribution, and some papers \cite{campi2008exact,geng2021computing} define it as sample complexity. We define it as solution complexity here to distinguish it from the complexity in the later scenario compression process.
\begin{definition}[Solution Risk] \label{vp}
The \textit{solution risk} (also called \textit{violation probability}) of a candidate solution $(g^*,\eta^*)$ is defined as the probability that $(g^*,\eta^*)$ is infeasible, i.e., $\mathbb{V}_{\tilde{w}}(g^*,\eta^*):=\mathbb{P}_{\tilde{w}}((g^*,\eta^*)\notin \mathcal{X}_{\tilde{w}})$, where $\mathcal{X}_{\tilde{w}}$ is the decision set generated by $\mathrm{SP}(\mathcal{N})$.
\end{definition}
\indent Based on the above definitions, Campi \textit{et al.} first gives the relationship between the risk and number of scenarios for $\mathrm{SP}(\mathcal{N})$ in \cite{calafiore2006scenario,campi2008exact}, and update it to a more applicable form recently in \cite{garatti2022risk} which reveals a fundamental correlation structure that links the risk to the solution complexity, i.e. Theorem \ref{theorem1}.
\begin{theorem}[Solution Risk Bounds \cite{garatti2022risk}] \label{theorem1} 
Consider the number of scenarios $N$ is larger than the number of decision variables $n$. Given a confidence parameter $\beta \in (0,1)$, for any $k = 0,1,\dots,n$ consider the polynomial equation in the $t$ variable
\begin{equation}
\resizebox{1\hsize}{!}{$
    \binom{N}{k} t^{N-k}-\frac{\beta}{2 N} \sum_{i=k}^{N-1} \binom{i}{k} t^{i-k}
    -\frac{\beta}{6 N} \sum_{i=N+1}^{4 N} \binom{i}{k} t^{i-k}=0
$}
\end{equation}
This equation has exactly two solutions in $[0,+\infty)$, which are named as $\underline{t}(k) <$. Supposing $\underline\epsilon(k):= max\{0,1-\bar{t}(k)\}$ and $\bar\epsilon(k):= 1-\underline{t}(k)$, it holds that
\begin{equation} \label{exact}
\mathbb{P}_{\tilde{w}}^N\left\{\underline{\epsilon}\left(s_N^*\right) \leq \mathbb{V}_{\tilde{w}}(g^*,\eta^*) \leq \bar{\epsilon}\left(s_N^*\right)\right\} \geq 1-\beta
\end{equation}
\end{theorem}
\indent Because the exact value of $s_N^*$ depends on the solution, the decision makers need to tune the input sample size $N$ based on solutions to finally hit the desired risk level. In the past ten years, many papers have addressed this issue, including how to tune sample size or risk efficiently \cite{campi2011sampling,garatti2022complexity}, and how to compute the value of $s_N^*$ in both convex and non-convex problems \cite{geng2021computing}. {\color{black} 
Meanwhile, the conservativeness of the risk upper bound has also significantly improved over the past decade. By setting a confidence parameter  $\beta=10^{-3}$ and a fixed scenario size $N=500$, Fig.\ref{complexity} illustrates the risk upper bounds provided by various studies. Notably, the most recent result from \cite{garatti2022risk} offers the least conservative risk bound, which we have applied for the first time in a power system problem.}
\begin{figure}[H]
\centering
\includegraphics[scale=0.42]{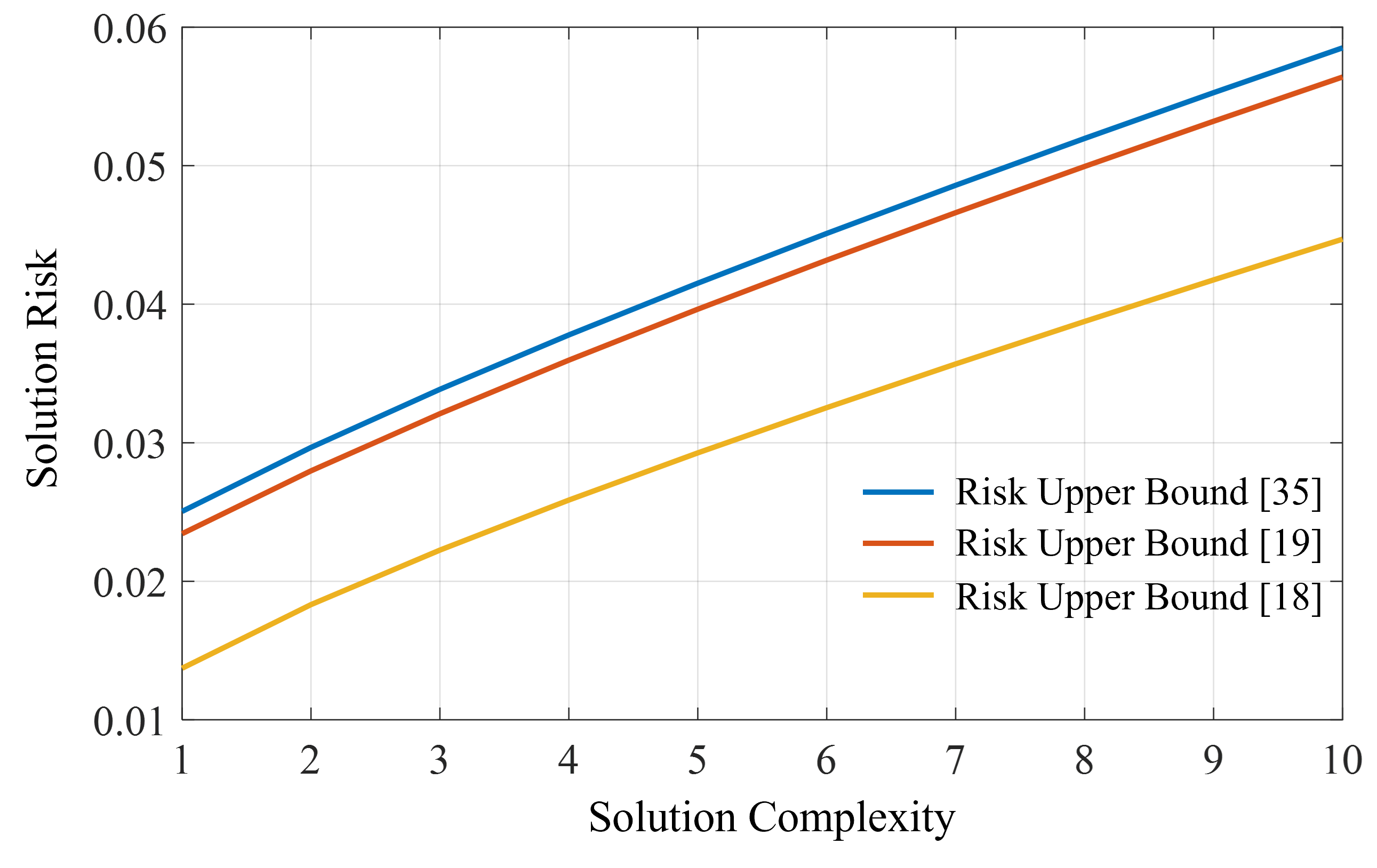}
  \caption{{\color{black}The comparison of risk upper bound from different papers}}
  \label{complexity}
\end{figure}
\indent Thanks to the convex property of the DC-OPF economic dispatch problem, the scenario-based problem can be easily solved by commercial software, such as CVX, CPLEX, and Gurobi. But when the system's original constraints and the scenario size are both very large, it will cause a long solving time. Meanwhile, even though the solving time may be improved in the future, only focusing on problem-solving time is not enough when applied to a large-scale power network for two reasons.\\ 
\indent First, the number of original constraints (\ref{saaa}) extend $N$ times after considering $N$ input scenarios, which increases the problem formulation time before solving it. Taking the IEEE 118 standard test system for example, the number of line, generator, and reserve constraints is 372, 108, and 108, respectively, which means the total number of constraints is $588 \times N$ after going through $N$ scenarios. Supposing we need 300 scenarios to meet the risk requirement, we must model a convex programming with 176400 constraints, which requires a huge computer memory and consumes much formulating time.\\ 
\indent Second, the frequent change of system operation modes, such as generators shouting down and line switching, means we might need to reformulate the programming. Similarly, tuning sample size $N$ to meet the desired risk also results in problem reformulating. The problem formulation is not a one-shot issue in scenario-based {\color{black}SCED} problems, and simplifying the problem formulation is crucial to applying the scenario approach to the real world.\\
\indent High-performance computing resource with large computer memory is helpful to store these constraints \cite{geng2021computing,modarresi2018scenario}. Still, the formulating time is hard to reduce because each constraint needs to be converted to standard form accepted by solvers \cite{lofberg2004yalmip,nethercote2007minizinc,legat2022mathoptinterface}. On the other hand, a simpler problem formulation can also save problem-solving time in return.   

\section{Scenario Compression} \label{sec:compre}
\indent Scenario compression means reducing the scenario size before formulating the problem while keeping the original solution as invariant as possible. In this section, the convex hull and box compression will be deeply discussed or proved, but the proposed compression risk validation scheme can be extended to any scenario compression method.
\subsection{Convex Hull Compression} \label{chc}
The convex hull compression means finding the smallest convex polygon, that encloses all of the scenarios in the set. {\color{black}This compression technique is similar to both deterministic optimization problems \cite{ardakani2013identification} and data-driven robust optimization problems \cite{wang2020data}, but it is first proved to be compatible with the scenario approach for the chance-constrained {\color{black}SCED} problem in this paper.}
\begin{definition}[Vertex of Convex Hull] \label{vertexofhull}
A scenario $\tilde{w}_v$ after convex hull compression is called \textit{vertex} if it cannot be written in the form $\tilde{w}_v = t\tilde{w}_x+(1-t)\tilde{w}_y$ where $\tilde{w}_x$ and $\tilde{w}_y$ of $\mathcal{N}$ are distinct scenarios and $0<t<1$, i.e. $\tilde{w}_v$ is not between two other points of $\mathcal{N}$.
\end{definition}
\indent Suppose $v$ vertexes labeled $\tilde{w}_{i_1}, \tilde{w}_{i_2},..., \tilde{w}_{i_v}$ after convex hull compression of the given scenarios set $\mathcal{N}$. Leveraging on the special structure of the robust counterpart, we can use the vertex of the convex hull to represent the whole uncertainty scenario-generated space, which is stated as Theorem \ref{vertex-all}.
\begin{theorem}[Invariant Solution]  \label{vertex-all}
The solution of the {\color{black}SCED} scenario problem $\mathrm{SP}(\mathcal{N})$ is invariant after replacing the scenarios set $\mathcal{N}$ with its convex hull vertexes $\tilde{w}_{i_1}, \tilde{w}_{i_2},..., \tilde{w}_{i_v}$.
\end{theorem}
\begin{proof}
Because the objective function and other deterministic constraints are the same before and after convex hull compression, proving the solution invariant property is equal to proving the constraints region $\mathcal{R}_1$ constructed by (\ref{saaa}) is the same as the region $\mathcal{R}_2$ constructed by (\ref{vertex})
\begin{subequations} \label{vertex}
\begin{align}
\underline{f} \preceq f(\hat{w}, \tilde{w}_i) \preceq \bar{f} \label{vertexa} \\ 
\underline{g} \preceq g-\mathbf{1}^{\top} \tilde{w}_i \eta \preceq \bar{g} \label{vertexb}\\ 
R_d\preceq-\mathbf{1}^{\top} \tilde{w}_i \eta\preceq R_u \label{vertexc}
\end{align}
\end{subequations}
with $i=i_1,i_2,\dots,i_v$.\\
\indent Considering the fact $\{\tilde{w}_{i_1}, \tilde{w}_{i_2},..., \tilde{w}_{i_v}\} \in \mathcal{N}$, it is easy to show $\mathcal{R}_1 \subseteq \mathcal{R}_2$. Now we need to prove $\mathcal{R}_1 \supseteq \mathcal{R}_2$.\\ 
\indent 1) For the generator capacity constraints (\ref{sab},\ref{vertexb}) and reserve constraints (\ref{sac},\ref{vertexc}), their corresponding constraints region $\mathcal{R}_1(\text{b,c}) \supseteq \mathcal{R}_2(\text{b,c})$ is equal to:
\begin{equation} 
\begin{array}{c}
-\mathbf{1}^{\top} \tilde{w}_j  \leq \text{max} \{-\mathbf{1}^{\top} \tilde{w}_{i_1}, \dots ,-\mathbf{1}^{\top} \tilde{w}_{i_v}   \}\\
-\mathbf{1}^{\top} \tilde{w}_j  \geq \text{min} \{-\mathbf{1}^{\top} \tilde{w}_{i_1}, \dots ,-\mathbf{1}^{\top} \tilde{w}_{i_v}\}\\
\end{array}
\end{equation}
for any $j \in \{1,2,...,N\}$.\\
\indent Note that $\{\tilde{w}_1,\dots,\tilde{w}_N\}$ is the convex combination of $\{\tilde{w}_{i_1},\dots,\tilde{w}_{i_v}\}$, which means any $\tilde{w}_j$ with $j \in \{1,2,...,N\}$ can be expressed as: 
\begin{equation} \label{combination}
\begin{aligned}
&\tilde{w}_j=\theta_{j_1}\tilde{w}_{i_1}+\theta_{j_2}\tilde{w}_{i_2}+...+\theta_{j_v}\tilde{w}_{i_v} \\
&\theta_{j_1} + \theta_{j_2} +...+\theta_{j_v}=1\\
&\theta_{j_1}, \theta_{j_2},...,\theta_{j_v} > 0
\end{aligned}
\end{equation}
Then we have:
\begin{equation}
\begin{aligned}
&-\mathbf{1}^{\top} \tilde{w}_j   = -\mathbf{1}^{\top} (\theta_{j_1}\tilde{w}_{i_1}+\theta_{j_2}\tilde{w}_{i_2}+...+\theta_{j_v}\tilde{w}_{i_v})  \\
=& -(\theta_{j_1}\mathbf{1}^{\top}\tilde{w}_{i_1}+\theta_{j_2}\mathbf{1}^{\top}\tilde{w}_{i_2}+...+\theta_{j_v}\mathbf{1}^{\top}\tilde{w}_{i_v}) \\
\leq & (\theta_{j_1} +...+\theta_{j_v}) \times \text{max} \{-\mathbf{1}^{\top} \tilde{w}_{i_1}, \dots ,-\mathbf{1}^{\top} \tilde{w}_{i_v}\}\\
=&\text{max} \{-\mathbf{1}^{\top} \tilde{w}_{i_1}, \dots ,-\mathbf{1}^{\top} \tilde{w}_{i_v}\}
\end{aligned}
\end{equation}
\indent Similarly, $-\mathbf{1}^{\top} \tilde{w}_j \geq \text{min} \{-\mathbf{1}^{\top} \tilde{w}_{i_1}, \dots ,-\mathbf{1}^{\top} \tilde{w}_{i_v}  \}$, which means we have proved:
\begin{equation} \label{Rbc}
\mathcal{R}_1(\text{b,c}) \supseteq \mathcal{R}_2(\text{b,c})
\end{equation}
\indent 2) For the line capacity constraints (\ref{saa},\ref{vertexa}), because $\mathbf{1} \tilde{w}$ is a scalar, we can extract the $\tilde{w}$ related robust counterpart as below:
\begin{equation}
\begin{aligned}
f(\hat{w}, \tilde{w})&=H_g\left(g-\mathbf{1}_{1 \times n_w} \tilde{w} \eta\right)-H_d d+H_w(\hat{w}+\tilde{w})\\
&= L \tilde{w}+(H_g g-H_d d+H_w\hat{w}) 
\end{aligned}
\end{equation}
where $L_{n_l \times n_w} = H_g\eta\mathbf{1}_{n_g \times n_w}+H_w$.\\
\indent Because $\eta$ is the decision variable, which is unknown before solving the problem. If $\mathcal{R}_1(\text{a}) \supseteq \mathcal{R}_2(\text{a})$ is true, supposing $L(r,:)$ represents the $r$-th row of the matrix $L$, we need to prove that for any $L$, there exists:
\begin{equation}
\begin{aligned}
L(r,:)\tilde{w}_j \leq \text{max} \{L(r,:)\tilde{w}_{i_1},\dots,L(r,:)\tilde{w}_{i_v}\} \\
L(r,:)\tilde{w}_j \geq \text{min} \{L(r,:)\tilde{w}_{i_1},\dots,L(r,:)\tilde{w}_{i_v}\} 
\end{aligned}
\end{equation}
for any $j \in \{1,2,...,N\}$ and $r \in \{1,2,...,n_l\}$.\\
\indent Actually, based on the convex combination result in (\ref{combination}), we can use a similar technique to show that:
\begin{equation}
\begin{aligned}
&L(r,:)\tilde{w}_j   = L(r,:) (\theta_{j_1}\tilde{w}_{i_1}+\theta_{j_2}\tilde{w}_{i_2}+...+\theta_{j_v}\tilde{w}_{i_v})  \\
=& (\theta_{j_1}L(r,:)\tilde{w}_{i_1}+\theta_{j_2}L(r,:)\tilde{w}_{i_2}+...+\theta_{j_v}L(r,:)\tilde{w}_{i_v}) \\
\leq & (\theta_{j_1}  +...+\theta_{j_v}) \times \text{max} \{L(r,:) \tilde{w}_{i_1}, \dots ,L(r,:) \tilde{w}_{i_v}\}\\
=&\text{max} \{L(r,:) \tilde{w}_{i_1}, \dots ,L(r,:) \tilde{w}_{i_v}\}
\end{aligned}
\end{equation}
\indent Similarly, $L(r,:)\tilde{w}_j \geq \text{min} \{L(r,:)\tilde{w}_{i_1},\dots,L(r,:)\tilde{w}_{i_v}\}$ is also true, which means we have proved:
\begin{equation}\label{Ra}
    \mathcal{R}_1(\text{a}) \supseteq \mathcal{R}_2(\text{a})
\end{equation}
\indent Because $\mathcal{R}_1 = \mathcal{R}_1(a) \cap \mathcal{R}_1(b,c)$ and $\mathcal{R}_2 = \mathcal{R}_2(a) \cap \mathcal{R}_2(b,c)$, after combing the results of (\ref{Rbc}) and (\ref{Ra}), it follows that:
\begin{equation}
    \mathcal{R}_1 \supseteq \mathcal{R}_2
\end{equation}
\end{proof}
\begin{corollary}[Invariant Solution Complexity] \label{isc}
For the {\color{black}SCED} scenario problem $\mathrm{SP}(\mathcal{N})$, the solution complexity $s_N^*$ is invariant after replacing the scenarios set $\mathcal{N}$ with its convex hull vertexes $\tilde{w}_{i_1}, \tilde{w}_{i_2},..., \tilde{w}_{i_v}$. 
\end{corollary}
\begin{proof}
Regarding the convex hull compression process as the removal of non-vertex scenarios from $\mathcal{N}$, based on the invariant solution Theorem \ref{vertex-all}, it shows that these non-vertex scenarios are non-support scenarios, i.e. the support scenarios are all kept after compression. Because the solution complexity $s_N^*$ is equal to the number of support scenarios, which means $s_N^*$ is also invariant after convex hull compression.
\end{proof}
{\color{black}Based on Theorem \ref{vertex-all} and Corollary \ref{isc}, the decision maker can simplify the problem formulation and reduce solving time by considering only the convex hull vertices instead of the entire scenario set in the scenario problem $\mathrm{SP}(\mathcal{N})$.}
\subsection{Box Compression}
After convex hull compression, the decision-makers only need to use the vertex scenarios to represent the whole scenario set $\mathcal{N}$ without solution change, but the number of vertex can be large under some uncertainties. Constructing a bigger compressed set with fewer vertexes is a straightforward way to further reduce the problem size.\\
\indent Without loss of generality, we introduce the box compression method which compresses the scenarios into a hyper-rectangle space \cite{margellos2014road}. Supposing $\tilde{w}_i^{(q)} \in [\underline\tau_q,\bar\tau_q]$ means the $q$-th element of the scenario $\tilde{w}_i$, the simple optimization problems (\ref{box}) relax the scenario set $\mathcal{N}$ to hyper-rectangle space $B^*:=\times_{q=1}^{n_w}\left[\underline{\tau}_q^*, \bar{\tau}_q^*\right]$
\begin{equation} \label{box}
    \begin{aligned}
\min _{\underline\tau_q,\bar\tau_q \in \mathbb{R}} & \left(\bar{\tau}_q-\underline{\tau}_q\right) \\
\text { s.t. } & \tilde{w}_i^{(q)} \in\left[\underline{\tau}_q, \bar{\tau}_q\right], \text { for } i=1, \ldots, N
\end{aligned}
\end{equation}
\indent Fig.\ref{compresscompare} illustrates the difference between convex hull compression and box compression. It is clear that the compressed space is relaxed to achieve the box compression shape with fewer vertexes. Because the hyper-rectangle after box compression is a special convex hull and all the scenarios are inside or on the boundaries of this hyper-rectangle, similar in Section \ref{chc}, we can also directly use the vertexes to represent the whole hyper-rectangle space.
\begin{figure}[H]
\centering
  \includegraphics[scale=1]{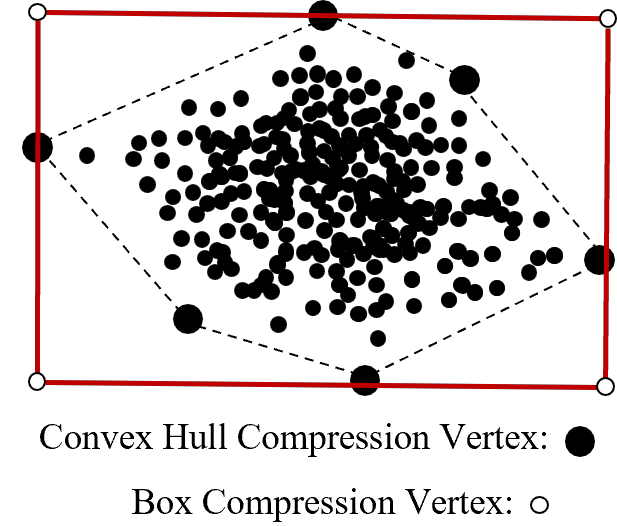}
  \caption{The comparison of convex hull compression (black dotted lines) and box compression (red solid lines)}
  \label{compresscompare}
\end{figure}
\indent Unlike convex hull compression, the vertex in box compression does not necessarily belong to the original scenario set $\mathcal{N}$, which means the exact value of solution complexity $s_N^*$ is impossible to be computed based on hyper-rectangle vertexes. To conquer this barrier, the authors in the previous paper replace $s_N^*$ with its upper bound: $n$ in (\ref{exact}), but it makes the final risk bound too conservative \cite{margellos2014road,vrakopoulou2013probabilistic}. In the next section, we will show that if it is hard to gain a precise solution risk, solving the compression risk may be another accessible way.
\subsection{Compression Risk}
\indent The above scenario compression can be regarded as applying an algorithm $\mathcal{A}$ to construct the corresponding compressed set of the sampled scenarios $\mathcal{N}$, and the compression function $\textbf{c}$ returns the smallest subset of $\mathcal{N}$ which can be used to reconstruct the compressed set. For the convex hull compression, function $\textbf{c}$ maps to the scenarios at the vertex, while for the box compression, function $\textbf{c}$ typically maps to the scenarios on the hyper-rectangle boundaries.\\
\indent Instead of focusing on the risk of the solution in Definition \ref{vp}, the compression risk directly quantifies the probability of a newly drawn scenario that is outside of the compressed set.  
\begin{definition}[Compression Risk \cite{campi2023compression}]\label{cr}
The \textit{compression risk} of the compression algorithm $\mathcal{A}$ is defined as the probability of change of vertexes after adding a new scenario:
\begin{equation} \label{crequ}
\begin{aligned}
&\phi(\mathcal{A}(\tilde{w}_1,...,\tilde{w}_N))  \\
= &\mathbb{P}_{\tilde{w}}\{\textbf{c}(\textbf{c}(\tilde{w}_1,...,\tilde{w}_N),\tilde{w}_{N+1}) \neq \textbf{c}(\tilde{w}_1,...,\tilde{w}_N)\}
\end{aligned}
\end{equation}
\end{definition}
where $\{\tilde{w}_1,...,\tilde{w}_N\}$ is the specific given data set, while $\tilde{w}_{N+1}$ is the future unknown data under distribution $\mathbb{P}_{\tilde{w}}$.
\indent Generally speaking, the compression risk only considers the risk in the sample space, while the solution risk projects the risk from the sample space to the solution space, which seems more useful for decision-makers. Under the scenario approach scheme, the risk in the solution space also means the risk in the sample space, but not vice versa. To illustrate this, suppose we have the simple chance-constrained problem (\ref{simple}):
\begin{equation}\label{simple}
\begin{aligned}
\min _{\theta,J}&\; J   \\
\text { s.t. }&\mathbb{P}_{\delta} \{ \rho(\theta,\delta) \leq J   \} \geq 1-\epsilon'
\end{aligned}
\end{equation}
where $\theta$ is the scalar decision variable, $\delta \in \Delta_\delta$ is the uncertain variable with two dimensions, $\rho$ is a simple convex function, $J$ is both the scalar decision variable and objective function, and $\epsilon'$ is the tolerable risk level.\\
\indent After applying the scenario approach to (\ref{simple}), we have the scenario-based formulation (\ref{simplesa}) if $N$ scenarios are needed.
\begin{equation}\label{simplesa}
\begin{aligned}
\min _{\theta,J}&\; J   \\
\text { s.t. }& \rho(\theta,\delta_i) \leq J  \quad\quad i=1,...,N
\end{aligned}
\end{equation}
\indent Assuming the solution of Equation (\ref{simplesa}) remains unchanged following convex hull compression, Fig.\ref{RC} illustrates the difference between the Definition \ref{vp} solution risk and Definition \ref{cr} compression risk. Any newly drawn scenario outside of the convex hull area is regarded as a risk of the compression process, but only part of these scenarios will cause violations in the solution space. For example, the red scenario $\delta_{N+1}'$ and the green scenario $\delta_{N+1}''$ are both viewed as a risk in compression, but only the red scenario $\delta_{N+1}'$ results in solution risk.
\begin{figure}[H]
\centering
  \includegraphics[scale=0.45]{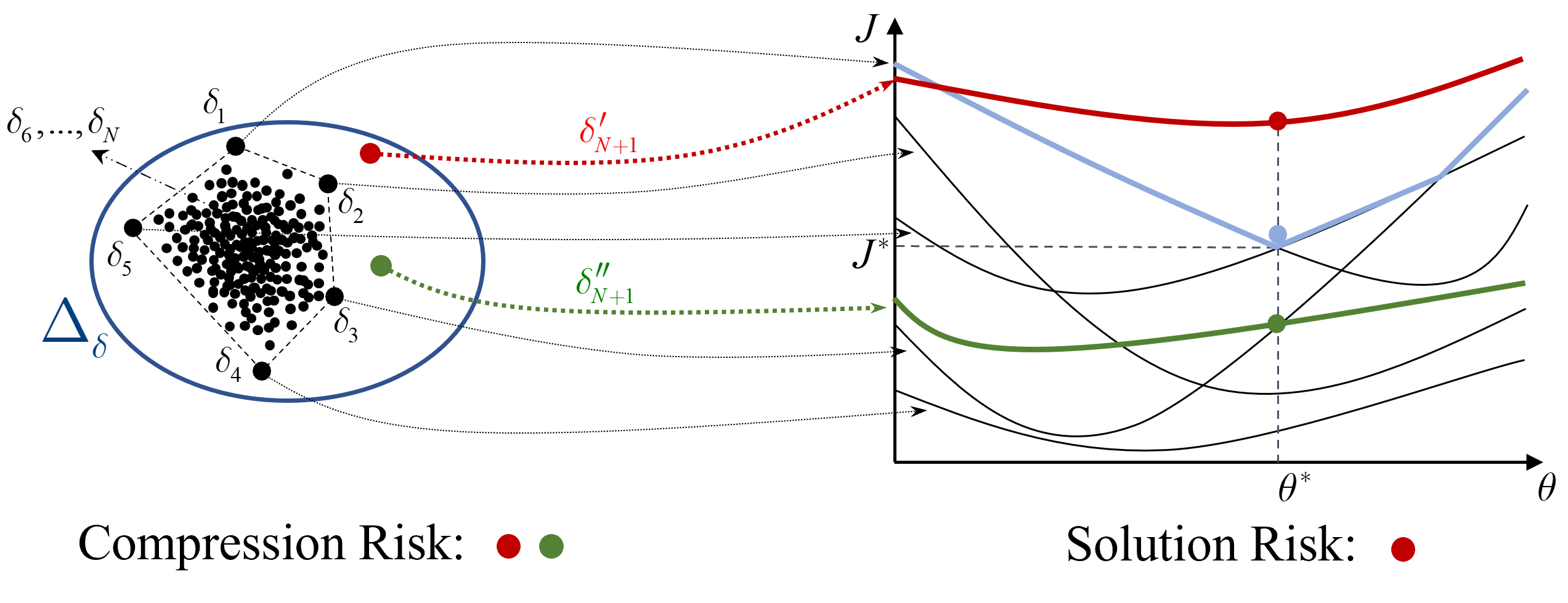}
  \caption{The comparison of (convex hull) compression risk and solution risk}
  \label{RC}
\end{figure}
\indent In Fig.\ref{RC}, the number of vertexes after compression is 5, while solution complexity $s_N^*$ is 2 equal to the number of decision variables. That means for the problem (\ref{simple}), it is easy to quantify the risk after projecting it from sample space to solution space.\\
\indent However, when the optimization problem is complex, it may be hard to compute solution complexity $s_N^*$ or the value of $s_N^*$ is very large. For a convex problem, the decision maker can replace $s_N^*$ with its upper bound, i.e. the number of decision variables $n$, but it is proved to be impractical in the {\color{black}SCED} problem, where the solution complexity $s_N^*$ is $3-5$ while the decision variables $n$ is 864 in \cite{modarresi2018scenario}.\\
\indent On the other hand, because of many mature compression algorithms, it is highly possible to quantify the compression risk bound from the sample space. Thanks to a recent theory breakthrough in \cite{campi2023compression}, the \textit{unprecedentedly} tight risk bounds based on the cardinality of the compressed data set are proved for the many compression processes. This result without requiring solution complexity $s_N^*$ is extremely useful when the sample space is simpler than the solution space.   
\begin{property}[Non-concentrated Mass] \label{nomass}
If the random variable $\boldsymbol{z} \in \mathcal{Z}$ satisfies:
\begin{equation}
\mathbb{P}\left\{\boldsymbol{z}=z\right\}=0, \forall z \in \mathcal{Z}
\end{equation}
then we say the random variable $\boldsymbol{z}$ has a non-concentrated mass property, which excludes with probability 1 that the same $z$ occurs twice or more times in a training set.
\end{property}
\begin{property}[Non-associativity] \label{noasso} 
For any $q \geq 1$, there exists
\begin{equation}
\mathbb{P}(E_1 	\setminus E_2)=0    
\end{equation}
where
\begin{equation}
\begin{aligned}
E_1=\{\textbf{c}(\boldsymbol{z}_1,\dots,\boldsymbol{z}_N,\boldsymbol{z}_{N+i})=\textbf{c}(\boldsymbol{z}_1,\dots,\boldsymbol{z}_N),i=1,\dots,q\} \\
E_2=\{\textbf{c}(\boldsymbol{z}_1,\dots,\boldsymbol{z}_N,\boldsymbol{z}_{N+1},\dots,\boldsymbol{z}_{N+q})=\textbf{c}(\boldsymbol{z}_1,\dots,\boldsymbol{z}_N)\}
\end{aligned}
\end{equation}
which means if the compression does not change after adding each new scenario one time, then it does not change by adding all new scenarios together.
\end{property}
\begin{property}[Preference] \label{pref}
For any multisets $U$ and $V$ such that $V \subseteq U$, if $V \neq \textbf{c}(U)$, then $V \neq \textbf{c}(U,z)$ for all $z \in \mathcal{Z}$. 
\end{property}

\begin{definition}[Compression Complexity] \label{comprecomple}
The compression complexity $k_c$ is defined as the minimum number of scenarios required to construct the compression set.
\end{definition}
Based on Definition \ref{comprecomple}, the compression complexity of both convex hull and box compression methods can be computed.
\begin{corollary}[Compression Complexity of Convex Hull Compression] \label{hullcomplex}
The compression complexity of convex hull compression is equal to the number of vertexes of the compression set.
\end{corollary}
\begin{proof}
According to Definition \ref{vertexofhull}, removing any convex hull vertex will alter the shape of the compression set. Consequently, the minimum number of scenarios required to construct the compression set is equal to the number of vertices.
\end{proof}
\begin{corollary}[Compression Complexity of Box Compression] \label{boxcomplex}
The compression complexity of box compression is equal to or less than $2m$, where $m$ is the number of the uncertainty domain dimension.
\end{corollary}
\begin{proof}
In contrast to convex hull compression, the box compression set has precisely $2^m$ vertices. These vertices are derived from the original scenarios that touch with the box. This implies that the minimum number of scenarios required to define the box set is equal to or less than $2m$.
\end{proof}

\begin{theorem}[Compression Risk Bounds \cite{campi2023compression}] \label{crb}
Give a compression algorithm $\mathcal{A}$, and the compression function $\textbf{c}$ returns the smallest subset of the input data. If the probability by which the scenarios are drawn has no concentrated mass (Property \ref{nomass}), and the compression function $\textbf{c}$ satisfies Property \ref{noasso} and Property \ref{pref}, then we have the upper and lower bounds for the compression risk:
\begin{equation} \label{ulb}
\mathbb{P}\left\{\underline{\varepsilon}(k_c) \leq \phi \left(\mathcal{A}\left(\tilde{w}_1, \ldots, \tilde{w}_N\right)\right) \leq \bar{\varepsilon}(k_c)\right\} \geq 1-\beta
\end{equation}
where $k_c=\left|\textbf{c}\left(\tilde{w}_1, \ldots, \tilde{w}_N\right)\right|$ while $\underline{\varepsilon}(k_c)$ and $\bar{\varepsilon}(k_c)$ has the same definition as Theorem \ref{theorem1}.
\end{theorem}
\indent The key difference between Theorem \ref{theorem1} and Theorem \ref{crb} is that the $s_N^*$ in (\ref{exact}) is number support scenarios based on solution, while $k_c$ in (\ref{ulb}) is the compression complexity of the compressed data set. {\color{black}As illustrated in Fig. \ref{RC}, the solution risk guarantee offers a tighter risk bound compared to the compression risk. If a more accurate risk metric (bound) is the decision maker's top priority, the solution risk is the preferable choice. However, if minimizing computing time is paramount, then the compression risk might become the better option. The balancing criteria can be designed to include indexes that account for both risk conservative and computation time.}\\
\indent \textit{Remark}: It is easy to check that the convex hull compression and box compression meet the requirement of Theorem \ref{crb}. If the decision maker wants to use other advanced learning-based compression algorithms, except for satisfying Property 1-3, the \textit{coherence} property of the loss function is also needed to get the risk bounds (\ref{ulb}) \cite{campi2023compression}. 

\subsection{Time Complexity}
\indent \textit{1) Convex Hull Compression}\\
\indent After convex hull compression, the number of constraints for scenario-based optimization is reduced from $N$ in (\ref{saaa}) to $v$ in (\ref{vertex}), which saves both time and computer memory in the problem formulation part. However, this trade-off process also sacrifices the computation time in finding the convex hull.\\
\indent Luckily, algorithms that construct convex hulls have been well-studied in the past decades. In the planar case, many algorithms have log-linear time complexity $O(N\log{N})$ \cite{de2000computational}, such as the Graham scan, Quickhull, and Chan's algorithm. For high dimension situations, the time complexity may increase to $O(N\log{N}+N^{\llcorner \frac{m}{2} \lrcorner})$ \cite{chazelle1993optimal}, {\color{black}which implies as the dimensionality of the uncertainty set increases, constructing a convex hull set becomes more time-consuming. To illustrate this, consider several independent uncertainty variables, each following a normal distribution with the same mean $\mu=0$ and standard deviation $\sigma = 0.02$, Figure \ref{convextime} compares the time required for convex hull compression across different dimensions.} 
\begin{figure}[H]
\centering
  \includegraphics[scale=0.5]{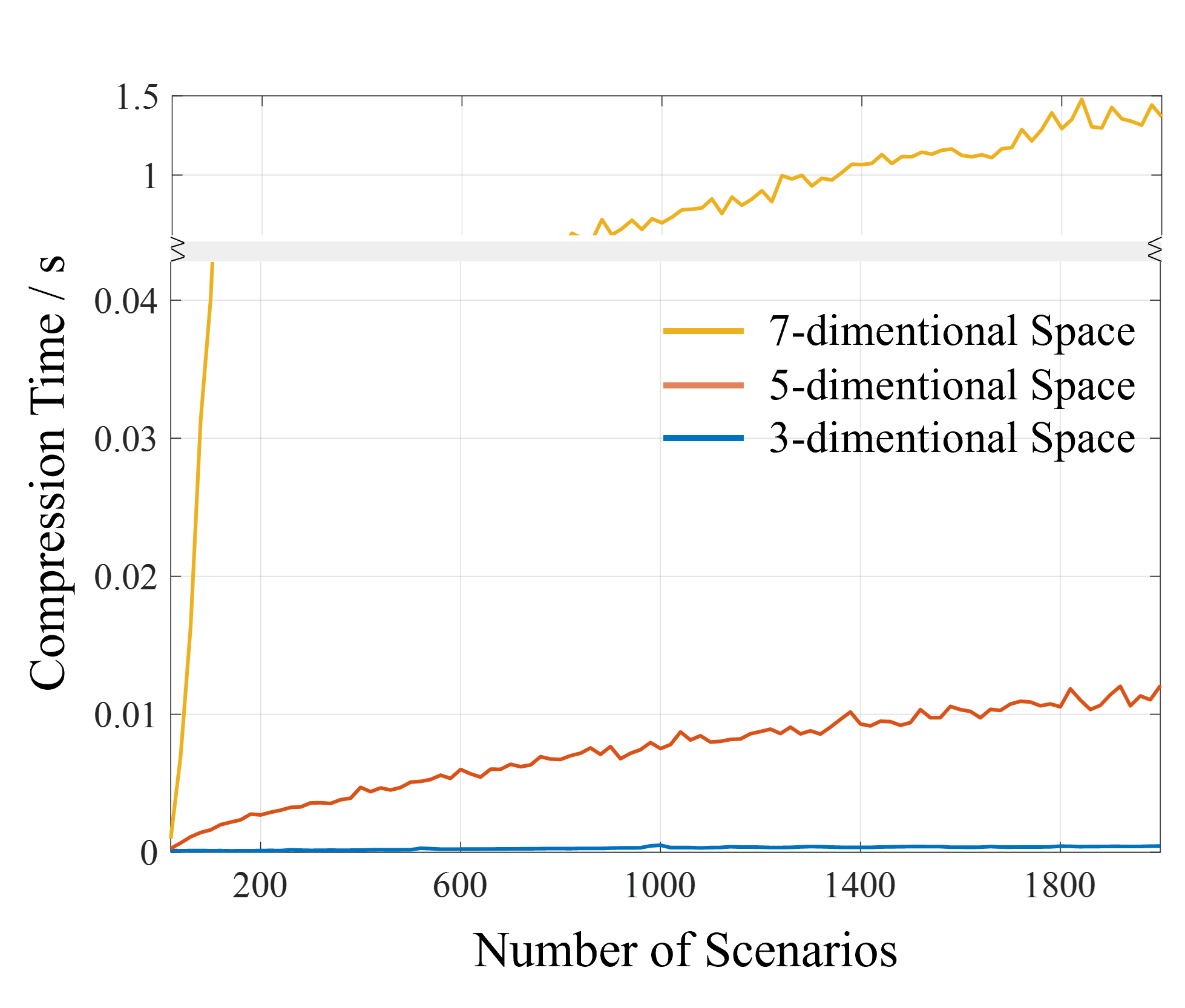}
  \caption{{\color{black}The convex hull compression time under different dimensionality}}
  \label{convextime}
\end{figure}
\indent {\color{black}However, in reality, elements within a scenario often exhibit correlations, which can significantly reduce solving time after model reduction is applied \cite{barber1996quickhull}.} Especially, if these elements have a linear correlation with each other, they can be projected to lower dimensional space by affine transformation.
\begin{definition}[Affine Transformation]
An affine transformation $F:\mathbb{R}^{n_a}\rightarrow\mathbb{R}^{m_a}$ is a map of the form $F(\mathbf{x})=A(\mathbf{x})+\mathbf{b}$ where $\mathbf{b}\in\mathbb{R}^{m_a}$ is a fixed vector and $A$ is an invertible linear transformation of $\mathbb{R}^{n_a}$.
\end{definition}
{\color{black}Let $C\subset\mathbb{R}^{n_a}$ be a convex set, then set $F(C)$ after affine transformation is also a convex set and $F$ maps the vertexes of $C$ onto the vertexes of $F(C)$ \cite{sutherland2009introduction}. Taking a 3-dimension random variable as an example, suppose two of its coordinate directions value has a linear relationship. Based on this property, we can generate these vertexes from a 2-dimension space rather than from the 3-dimension space, which is illustrated in Fig.\ref{proj}.}\\
\indent \textit{Remark}: Indeed, all supporting scenarios in Definition \ref{supportconst} are vertexes of convex hull. Hence, the solution complexity cannot be reduced by using convex hull compression, and the vertexes are usually more than the solution complexity. However, convex hull compression allows decision-makers to reduce the number of constraints in optimization and, therefore, leads to saving computational time when convex hull construction is easier than finding support scenarios, which is typically true in large-scale power system problems.
\begin{figure}[H]
\centering
  \includegraphics[scale=0.25]{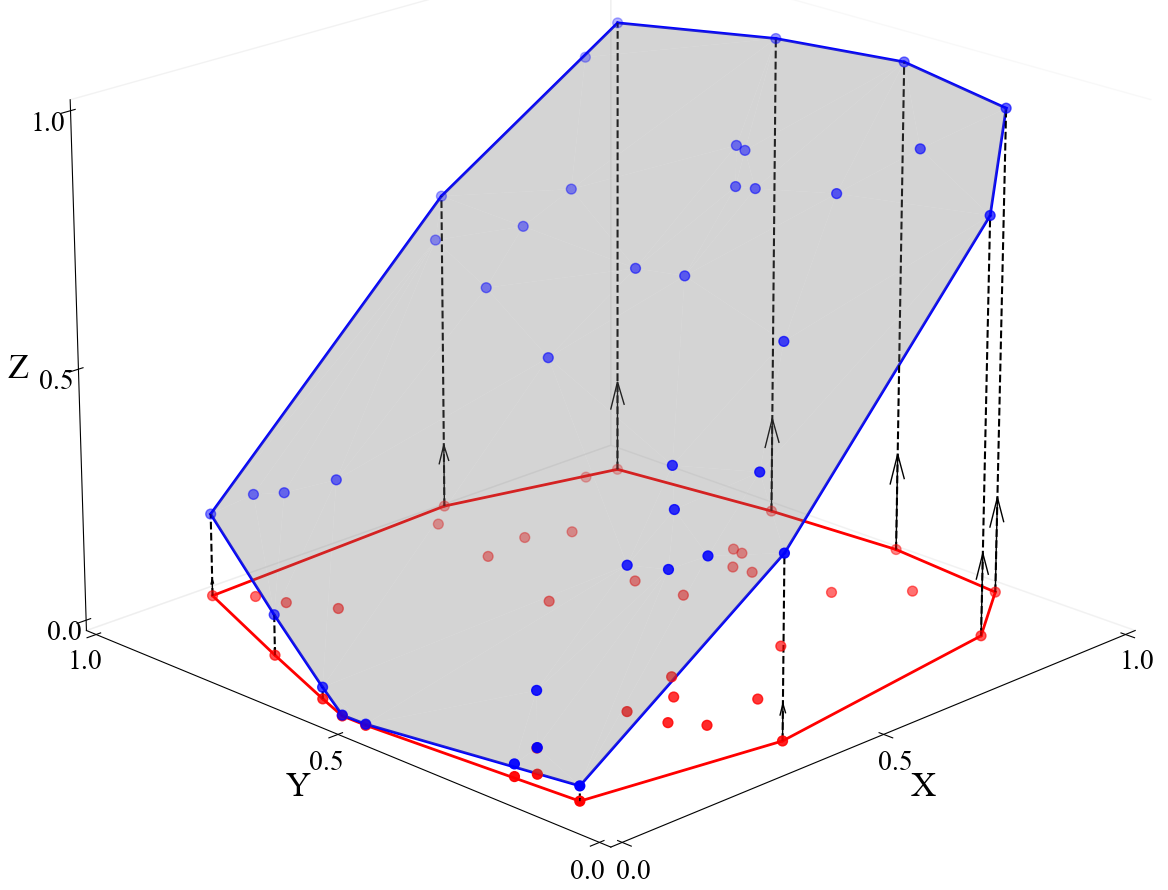}
  \caption{The simple illustration of affine transformation from black points to red points, where the black random variables are uniformly distributed in [0, 1] of each coordinate direction but with the same value in X-direction and Z-direction.}
  \label{proj}
\end{figure}
\indent \textit{2) Box Compression} \\
\indent {\color{black}Box compression can be conducted after convex hull compression, which extremely reduces the $N$ scenarios to $2^m$ scenarios. Meanwhile, directly constructing the box set from the original scenario set is faster, which equals to finding the minimum and maximum value of the scenario's dimension, whose time complexity is $O(mN)$.}\\
\indent The number of scenarios after compression $2^m$ might be a big value in a high-dimension case. Except for choosing a similar projection method when some elements of the scenario have a correlation, there are some other tractable reformulations that can be used to reduce the problem size. For example, if the generations cost $c(g)$ is a linear function and there are $n_c$ constraints in the chance-constraints part (\ref{subeqn:risk}), based on the results in \cite{bertsimas2006tractable}, the $n_c \times 2^m$ constraints after box compression can be further reduced to $n_c(2m+1)$ after adding $n_c(m+1)$ decision variables.

\section{Case Study} \label{sec:case}
\indent The similar forecasting error pattern of the spatial near wind farms is a well-recognized phenomenon \cite{xie2010wind,xie2013short}. For example, the Electric Reliability Council of Texas (ERCOT) divides wind farms into five geographical regions in which climatological characteristics are similar for all areas within such region, and the regional forecasting results are updated every five minutes online \cite{ercot}. To simplify the simulation, the relative forecasting error value of each wind farm from the same area is regarded as the same. In reality, the operators may use linear correlation to model them based on more detailed data, which makes the compression process faster.\\
\indent All the problems are solved using 64 GB memory on the Intel XEON CPU. The mathematical models were formulated using YALMIP on Matlab R2023a and solved using Gurobi v9.5. 
\subsection{118-bus System}
The 118-bus system is based on the test case \textit{c118swf.m} in MATPOWER \cite{murillo2013secure}. The grid includes 4 areas, 118 nodes, 210 lines, and 52 generators, 11 of which are modeled as wind farms located in 3 areas with 5, 4, and 2 farms in Area 1, Area 2, and Area 3 respectively. To address the influence of wind uncertainty on economic dispatch, the storage units are turned off during the simulation.\\
\indent  Suppose the wind relative forecasting errors have Normal distribution with mean $\mu=0$, while standard deviation $\sigma = 0.02$, $\sigma = 0.04$, and $\sigma = 0.06$ in Area 1,2, and 3 respectively. The absolute forecasting error scenarios of each wind farm can be easily affine projected from these three random variables by multiplying with their deterministic forecasting result $\hat{w}$, which is set to 0.25 of each wind farm's capacity in our simulation.\\
\indent After generating 500 independent scenarios extracted from the above uncertainty set, we apply the conventional scenario approach, convex hull compression, and box compression to solve the chance-constrained {\color{black}SCED} problem (\ref{cco}). {\color{black}Because of the randomness property of the generated data, the solution complexity $s_N^*$ might be slightly different for different scenario sets. To avoid this variation, complexity and risk are based on the average solution complexity and risk value of 100 times independent experiment.}\\
\indent {\color{black}We first evaluate the solution risk bounds provided by Theorem \ref{theorem1} against those from previous literature, as detailed in TABLE \ref{compsr}. The results demonstrate that the solution risk bounds from Theorem \ref{theorem1} are less conservative compared to other methods.}
\begin{table}[H]
\caption{{\color{black}The comparison of solution risk bounds}}
\centering
\label{compsr}
\begin{tabular}{{>{\color{black}}c >{\color{black}}c>{\color{black}}c}}
\toprule
$N=500$, $\beta=0.001$         & \textbf{Complexity} & \textbf{Solution Risk}      \\ \midrule
\textbf{Theorem 1} & 2.25       & {[}0, 0.031{]} \\
\textbf{\cite{campi2018wait,modarresi2018scenario}}   & 2.25       & {[}0, 0.029{]} \\
\textbf{\cite{campi2011sampling,ming2017scenario}}    & 2.25       & {[}0, 0.019{]}\\  \bottomrule
\end{tabular}
\end{table}
\indent {\color{black}Similarly, the compression risk bounds from Theorem \ref{crb} are compared and found to be less conservative than those from other approaches, as presented in TABLE \ref{compcr}.}
\begin{table}[H]
\caption{{\color{black}The comparison of box compression risk bounds}}
\centering
\label{compcr}
\begin{tabular}{{>{\color{black}}c >{\color{black}}c>{\color{black}}c}}
\toprule
$N=500$, $\beta=0.001$         & \textbf{Complexity} & \textbf{Compression Risk}   \\ \midrule
\textbf{Theorem 3} & 6          & {[}0, 0.045{]} \\
\textbf{\cite{margellos2014road,margellos2015connection}}    & 6          & {[}0, 0.033{]}\\  \bottomrule
\end{tabular}
\end{table}
\indent {\color{black}Utilizing the most accurate risk bounds available—those from Theorem \ref{theorem1} and Theorem \ref{crb}—the differences between solution risk and compression risk are highlighted in TABLE \ref{result_118bus}.} Compared with the conventional scenario approach, the convex hull compression saves over 10 times less time for problem formulation and solving without changing the solution. Meanwhile, box compression is shown to be the most efficient method, but it changes the original solution which means a higher dispatch cost. As illustrated in Fig.\ref{RC}, when the decision maker can easily compute the solution complexity $s_N^*$, the solution risk is a more precise risk metric than compression risk, such as the convex hull compression. {\color{black}However, when employing box compression or another method that complicates obtaining the exact solution complexity $s_N^*$, it becomes necessary to compare these methods with alternative compression techniques to estimate the solution risk. For example, box compression is typically considered less risky than convex hull compression.}
\begin{figure}[H]
\centering
  \includegraphics[scale=0.3]{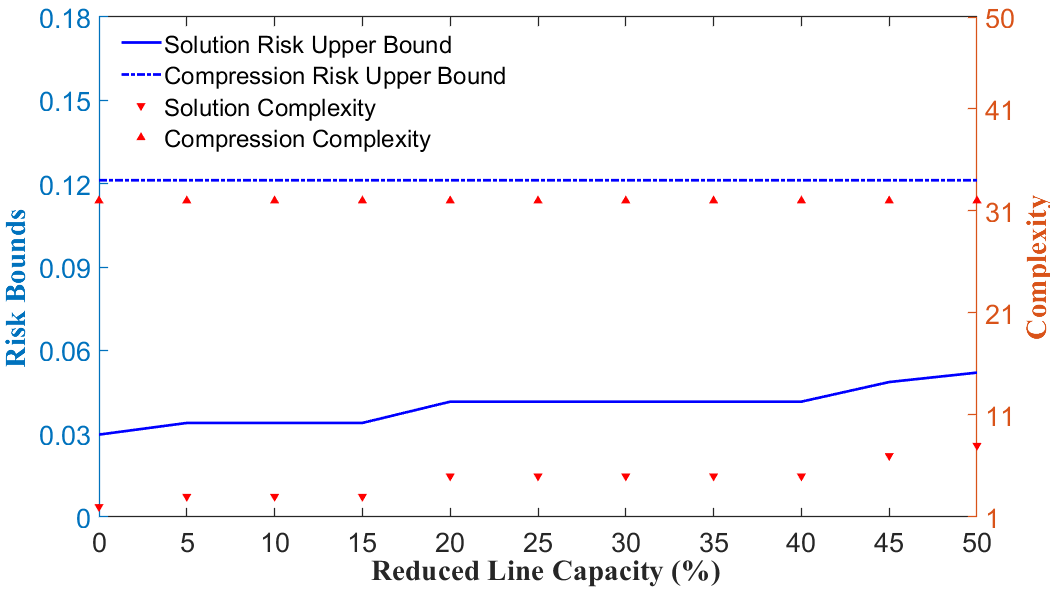}
  \caption{The affection of line capacity to the complexity and risk}
  \label{reducedline}
\end{figure}
\indent The gap between solution risk and compression risk is mainly decided by the complexity of uncertainty in both sample and solution space, which is measured by $k_c$ and $s_N^*$ respectively. As shown in Fig.\ref{reducedline}, after reducing the line capacity of the system in convex hull compression, the value of $s_N^*$ also increases because of more active constraints. However, the complexity of uncertainty in sample space $k_c$ is not affected under the same scenario set, which means the gap between these two risks is narrower when grids have less transmission line capacity or operate under a higher power demand situation.\\
\begin{table*}[htbp]
\caption{The comparison of different compression methods (118-bus System)}
\centering
\label{result_118bus}
\begin{threeparttable}
\begin{tabular}{cccccccc}
\toprule
$N=500$, $\beta=0.001$  & \textbf{Dispatch Cost} & \textbf{$s_N^*$} & \textbf{Solution Risk} & \textbf{$k_c$}      & \textbf{Compression Risk}   & \textbf{Formulating Time}\tnote{*} & \textbf{Solving Time}\tnote{**}   \\\midrule
\textbf{Without Compression}                   & $8.278 \times 10^4 \$ $    & 2.25 & {[}0, 0.031{]}  & / & /                  & 2.426s           & 4.494s         \\
\textbf{Convex Hull Compression}             & $8.278\times10^4 \$ $    & 2.25 & {[}0, 0.031{]} & 32 & {[}0.028, 0.121{]} & 0.162s           & 0.368s        \\
\textbf{Box Compression}                     & $8.279\times10^4 \$ $    & /  &  $<$ 0.031 & 6 & {[}0, 0.045{]} & 0.042s           & 0.141s      \\  \bottomrule
\end{tabular}
\end{threeparttable}
\end{table*}
\begin{table*}[htbp]
\centering
\caption{The comparison of different compression methods (Synthetic Texas Grids)}
\label{result_texas}
\begin{threeparttable}
\begin{tabular}{cccccccc}
\toprule
$N=500$, $\beta=0.001$  & \textbf{Dispatch Cost} & \textbf{$s_N^*$} & \textbf{Solution Risk} & \textbf{$k_c$}      & \textbf{Compression Risk}   & \textbf{Formulating Time}\tnote{*} & \textbf{Solving Time}\tnote{**}  \\\midrule
\textbf{Without Compression}                   & $ 1.0274\times10^6 \$ $    & 2 & {[}0, 0.030{]}  & / & /                  & 47.53s           & 136.60s         \\
\textbf{Convex Hull Compression}             & $1.0274\times10^6 \$ $     & 2 & {[}0, 0.030{]} & 63 & {[}0.071, 0.200{]} & 6.05s           & 16.96s        \\
\textbf{Box Compression}                     & $1.0275\times10^6 \$ $    & /  & {\color{black} $<$ 0.030} & 10 & {[}0.003, 0.059{]} & 3.14s           & 6.84s      \\  \bottomrule
\end{tabular}
\begin{tablenotes}
\item[*] {\color{black}The scenario compression time is included in the Formulating Time.}
\item[**] {\color{black}The practical requirement for the combined formulation and solving time is less than 120 seconds \cite{barry2022risk}.}
\end{tablenotes}
\end{threeparttable}
\end{table*}
\subsection{Synthetic Texas Grid}
\begin{figure}[H]
\centering
  \includegraphics[scale=0.55]{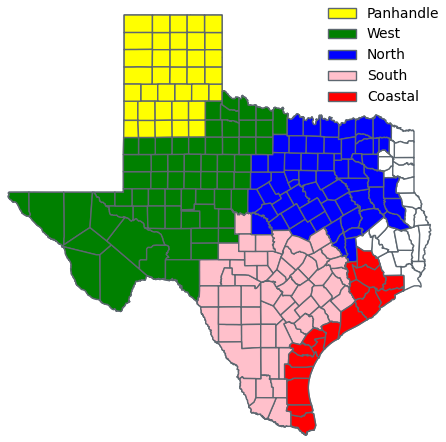}
  \caption{Wind forecasting regions in ERCOT}
  \label{forecastregion}
\end{figure}
\indent We compare different compression methods on a 2000-bus synthetic grid on a footprint of Texas \cite{birchfield2016grid}. This grid consists of 544 generation units, with a portfolio of 367 gas, 39 coal, 4 nuclear, 25 hydro, 87 wind, and 22 utility-scale solar power plants. Nodes with wind farms are where uncertainty exists in this paper, but the results can be generalized to solar and other uncertain sources in the electricity market. The wind power forecasting profiles are scaled from the real recorded data from ERCOT based on the five wind forecasting regions in Fig.\ref{forecastregion}. To avoid repeated simulation, we select the dispatch interval with average wind power in 2022 as the tested case which is 11:00 am on March 19th, whose wind output accounts for $32.5\%$ of the whole wind capacity \cite{ercot}.\\ 
\indent Unlike the assumed normal distribution function in the 118-bus System, the wind forecasting error scenarios for the synthetic Texas grids are directly generated from past experience, which unleashes the distribution-agnostic advantage of the scenario approach. To be precise, the scenario generation method is based on filtering scenarios using ambient wind conditions proposed in \cite{zhang}, which means the scenarios are generated from the past error data with the most similar ambient conditions, such as wind speed and temperature. Fig.\ref{3.19} shows the empirical distribution of the generated 500 scenarios for the studied dispatch interval, which is more precise than a naive normal distribution.
\begin{figure}[H]
\centering
  \includegraphics[scale=0.35]{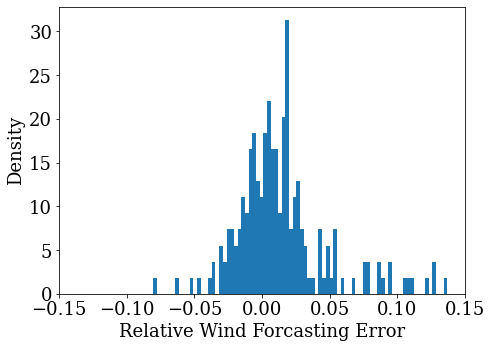}
  \caption{The empirical distribution of the generated wind forecasting error scenarios for the south Texas area on 11:00-11:05 dispatch interval on 3/19/2022}
  \label{3.19}
\end{figure}
\indent Because the synthetic Texas grid includes 3206 transmission lines, directly applying the compressed scenarios to each line constraint still makes the optimization problem consume a large memory and long solving time. To overcome this barrier, we first solve the conventional DC-OPF problem with deterministic forecasting value. Based on the deterministic results, only the line whose power flow exceeds $60\%$ of its capacity will be considered for the chance-constraints problem. After this key line selection process, the number of transmission lines in the scenario approach is reduced to 247 even less than the 118-bus system.\\ 
\indent The transmission line constraint reduction is similar to finding the \textit{active constraint} of the optimization problem, but it needs to make sure the discarded line constraints remain inactive for the uncertain scenarios. Considering the 5-min wind forecasting error is typically not large compared with long-term forecasting, we simply choose $60\%$ of the line capacity as the threshold to pick key line constraints, and the operator can change this threshold value based on the accuracy of the forecasting or the penetration level of wind power. A more precise but complicated method to identify the potential \textit{active constraint} sets can be found in \cite{misra2022learning}.\\
\indent TABLE \ref{result_texas} compares the results of applying different compression methods in synthetic Texas grids. As the system expands, the scenario-based problem formulation and solving time both increase, which in total is 184.13s without compression. Because each dispatch interval is 5 minutes, after including the delay in communication, control, and other factors, the real-time economic dispatch problem must be solved within 2 minutes \cite{barry2022risk}, which is hard to achieve in large scenario-based problems without compression. On the other hand, the problem formulation and solving time are both reduced extremely after convex hull or box compression.\\
\begin{figure}[H]
\centering
  \includegraphics[scale=0.33]{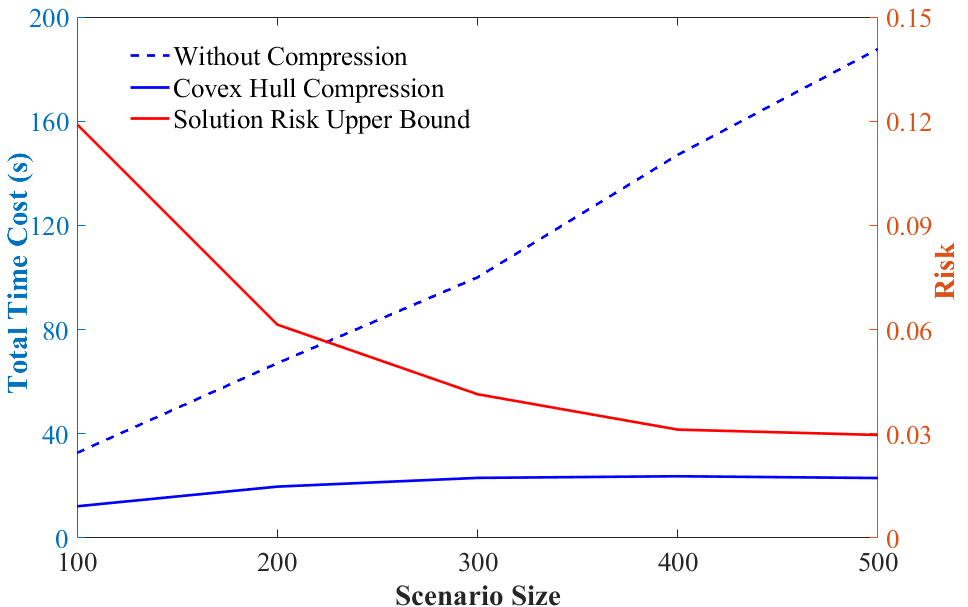}
  \caption{The total time cost based on different scenario sizes in the risk-tuning process}
  \label{risktuning}
\end{figure}
\indent In reality, when the operators are not satisfied with the risk bound after solving the problem, they need to change the number of input scenarios, which means repeating the problem formulation and solving process at least once. Fig.\ref{risktuning} illustrates the total time cost in the risk-tuning process based on different sample sizes. It is clear that the time gap between with and without convex hull compression is increasing as more scenarios are input to achieve a lower risk guarantee. As the operators tune the risk more times, the benefit of time-saving after scenario compression is also magnified.
  
\section{Conclusion} \label{sec:conclusion}
In this paper, scenario compression methods are proposed to speed up the chance-constrained {\color{black}SCED} problem. The convex hull compression using only vertices yields an equivalent solution to using the entire scenario set. Meanwhile, the proposed compression risk validation scheme helps assess the impact of any compression method within uncertainty domains. Depending on the complexity of the problem at hand, decision-makers have the flexibility to choose between solution risk and compression risk as the appropriate metric.\\
\indent {\color{black}Both convex hull and box compression techniques achieved nearly a tenfold reduction in the formulation and solving time on the 118-bus system and synthetic Texas grids, compared to scenarios without compression. This efficiency improvement paves the way for the application of chance-constrained methods in real-time economic dispatch, which requires solving problems within the (5-minute) dispatch interval.}\\
\indent The compression risk is highly related to the construction method of the compressed set in the uncertainty domains. {\color{black}Future work will explore optimal compression region constructions and apply our compression analysis scheme to other stochastic programming problems.}

\section{Acknowledgment}
The authors would like to thank Prof. M.C. Campi, Prof. S. Garatti{\color{black}, and anonymous reviewers} for their insightful suggestions.
\bibliographystyle{IEEEtran}
\bibliography{ref}

\begin{IEEEbiography}[{\includegraphics[width=1in,height=1.25in,clip,keepaspectratio]{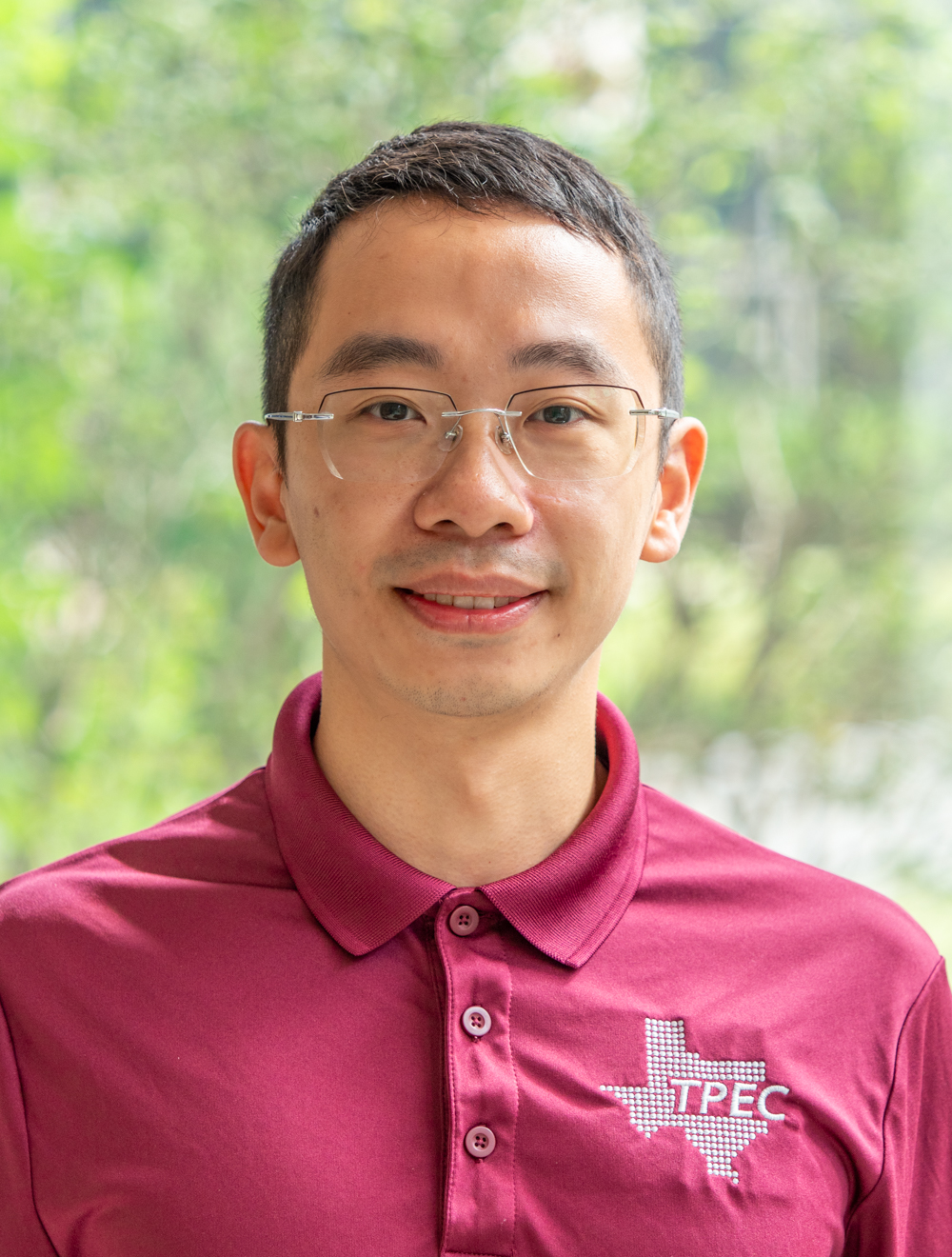}}]{Qian Zhang}
(Student Member, IEEE) received his B.E. and M.S. degrees in Electrical Engineering from Zhejiang University, Hangzhou, China, in 2019 and 2022 respectively. He is currently working toward the Ph.D. degree at Harvard University, Cambridge, MA, USA. His research interests include machine learning, optimization in the electricity market, and power system stability and control.  
\end{IEEEbiography}

\begin{IEEEbiography}[{\includegraphics[width=1in,height=1.25in,clip,keepaspectratio]{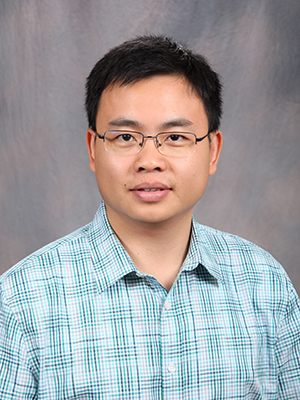}}]{Le Xie}
(Fellow, IEEE) received the B.E. degree in electrical engineering from Tsinghua University, Beijing, China, in 2004, the M.S. degree in engineering sciences from Harvard University, Cambridge, MA, USA, in 2005, and the Ph.D. degree from the Department of Electrical and Computer Engineering, Carnegie Mellon University, Pittsburgh, PA, USA, in 2009. He is currently a Professor at the School of Engineering and Applied Sciences, Harvard University, Cambridge, MA, USA. His research interests include modeling and control of large-scale complex systems, smart grid applications with renewable energy resources, and electricity markets.
\end{IEEEbiography}
\end{document}